\documentclass[bj, preprint]{imsart}

\usepackage{amsthm,amsmath,natbib, graphicx}
\RequirePackage[colorlinks,citecolor=blue,urlcolor=blue]{hyperref}

\startlocaldefs
\usepackage{amsthm,amsmath,amssymb,enumerate,euscript,mathrsfs,amsfonts}
\newtheorem{theorem}{Theorem}
\newtheorem{proposition}[theorem]{Proposition}

\newcommand{\eref}[1]{(\ref{#1})}

\newcommand{\sref}[1]{Section~\ref{#1}}
\newcommand{\tref}[1]{Table~\ref{#1}}
\newcommand{\fref}[1]{Figure~\ref{#1}}
\newcommand{\cref}[1]{Chapter~\ref{#1}}

\newcommand{\propref}[1]{Proposition~\ref{#1}}
\newcommand{\thmref}[1]{Theorem~\ref{#1}}

\def\bbB{\mathbb{B}}

\def\bbI{\mathbb{I}}

\def\bbL{\mathbb{L}}

\def\bbN{\mathbb{N}}

\def\bbP{\mathbb{P}}
\def\bbQ{\mathbb{Q}}
\def\bbR{\mathbb{R}}

\def\bbW{\mathbb{W}}

\def\bbZ{\mathbb{Z}}

\def\cL{{\mathcal L}}

\def\cS{{\mathcal S}}

\def\ga{\alpha}
\def\gb{\beta}

\def\gd{\delta}

\def\ggg{\gamma} 

\def\gk{\kappa}
\def\gl{\lambda}

\def\go{\omega}

\def\gr{\rho}
\def\gs{\sigma}
\def\gt{\tau}

\def\gU{\Upsilon}

\def\ds{\ensuremath{\displaystyle}}
\def\B{\ensuremath{\gamma}}	

\def\yb{\ensuremath{\bar{y}}}
\def\zb{\ensuremath{\bar{z}}}
\endlocaldefs

\begin{document}

\begin{frontmatter}

\title{Exact simulation of the sample paths of a diffusion with a finite entrance boundary}

\runtitle{Simulation of diffusions with a finite entrance boundary}

\begin{aug}

\author{\fnms{Paul A.} \snm{Jenkins}\thanksref{a}\corref{}\ead[label=e1]{p.jenkins@warwick.ac.uk}}
\address[a]{Department of Statistics, University of Warwick, Coventry CV4 7AL, United Kingdom.
\printead{e1}}

\runauthor{Jenkins}

\affiliation{Department of Statistics, University of Warwick, United Kingdom}

\end{aug}

\begin{abstract}
Diffusion processes arise in many fields, and so simulating the path of a diffusion is an important problem. It is usually necessary to make some sort of approximation via model-discretization, but a recently introduced class of algorithms, known as the \emph{exact} algorithm and based on retrospective rejection sampling ideas, obviate the need for such discretization. In this paper I extend the exact algorithm to apply to a class of diffusions with a finite entrance boundary. The key innovation is that for these models the Bessel process is a more suitable candidate process than the more usually chosen Brownian motion. 
The algorithm is illustrated by an application to a general diffusion model of population growth, where it simulates paths efficiently, while previous algorithms are impracticable.
\end{abstract}




\end{frontmatter}

\section{Introduction}

Diffusion processes arise in many important fields, including finance, genetics, and engineering \citep{klo:pla:1999}. There is great interest in simulation and inference using diffusions, but this is a difficult problem because the transition density function of a diffusion is rarely known. It is typical to work with a discretization of an intractable diffusion model so that Monte Carlo simulation can be applied. For example, the Euler scheme 
is a discretization of time in which increments of the diffusion over each small time step are assumed to be Gaussian. There has been much work into improving this and other discretized approaches \citep{klo:pla:1999}, but a disadvantage 
is that they introduce two sources of error: a Monte Carlo error and a discretization error. The latter causes a bias, and it may be computationally expensive to make the grid spacing sufficiently fine to ensure this bias is negligible. However, for a certain class of diffusion processes described below, recent work based on \emph{retrospective} sampling ideas has obviated the need for discretization, allowing realizations to be simulated \emph{exactly} \citep{bes:rob:2005, bes:etal:2006:B, bes:etal:2008}. The key idea is to use Brownian motion as the proposal in a rejection sampling algorithm, in which it is possible to make the accept/reject decision without having to simulate a complete (infinite-dimensional) sample path. This idea has also been extended 
to perform 
parametric inference of discretely observed diffusions 
\citep{bes:etal:2006:JRSSB, 
fea:etal:2008, fea:etal:2010}.

Where the algorithm has been developed for one-dimensional diffusions, it is assumed for the most part that the diffusion has state space $\bbR$ with 
boundaries at $\pm\infty$. The goal of this paper is to extend the exact algorithm to diffusion models with a certain class of \emph{finite} boundary. 
One approach would be a straightforward modification of the algorithm presented in \citet{bes:etal:2006:B} using Brownian motion as the candidate process. However, it is easy to find examples for which this approach exhibits unacceptably high rejection rates, essentially because the paths of Brownian motion do not mimic the paths of the target diffusion sufficiently well near the boundary. An alternative approach developed in this paper is more drastic and based on a related idea in \citet{sch:etal:2013}: we replace Brownian motion in the proposal mechanism with \emph{another} well-characterized diffusion, namely the Bessel process.
As a simple motivating example, 
suppose we wanted to simulate sample paths of the diffusion with generator 
\begin{equation}
\label{eq:Besseldriftgen}
\cL = \left[\frac{2\nu + 1}{2x} + \gr \frac{I_{\nu+1}(\gr x)}{I_\nu(\gr x)}\right]\frac{d}{dx} + \frac{1}{2}\frac{d^2}{dx^2},
\end{equation}
for $\gr > 0$ and $\nu \geq -1$, where $I_\nu(\cdot)$ is the modified Bessel function of the first kind. This diffusion was introduced by \citet{wat:1975:ZWVG} and studied by \citet{rog:pit:1981}, \citet{pit:yor:1981}, and \citet{sal:1984}, among others. 
The diffusion has state space $[0,\infty)$, and when $\nu \geq 0$ the boundary at 0 is \emph{entrance} in the terminology of Feller \citep[see][\S 15.6]{kar:tay:1981}. Thus, unlike Brownian motion which can wander close to 0, this diffusion should experience a strong repulsive force whenever the diffusion approaches 0 from above, due to the singularity in the drift of asymptotic form $\sim (2\nu + 1)/(2x)$ as $x \to 0$. 
Hence, a rejection sampling algorithm using Brownian motion as a candidate process will have high rejection rates whenever the sample path approaches the boundary. However, a \emph{Bessel} process of order $\nu$ shares this singularity 
and would make a more suitable candidate process. (In fact, in this example \eref{eq:Besseldriftgen} includes the Bessel process itself as the special case $\gr \to 0$. The diffusion with generator \eref{eq:Besseldriftgen} has been described as the \emph{Bessel process in the wide sense} \citep{wat:1975:ZWVG}.) 
As will be shown, using a Bessel process when the diffusion of interest has a certain type of boundary can substantially improve algorithmic efficiency. 

The structure of the paper is as follows. In \sref{sec:EA} I give an overview of the exact algorithms EA1--3 of \citet{bes:rob:2005}, \citet{bes:etal:2006:B}, and \citet{bes:etal:2008}. After summarizing some useful properties of the Bessel process in \sref{sec:Bessel}, in \sref{sec:Bessel-EA} I develop a new exact algorithm that can be applied to a diffusion with a finite entrance boundary. The algorithm is illustrated in \sref{sec:applications} by application to conditioned diffusions and to a general diffusion model of population growth. In \sref{sec:twoboundaries} I discuss the problem of \emph{two} finite boundaries and extend EA3 to apply to this case. \sref{sec:discussion} discusses possible directions for future research.

\section{Overview of the exact algorithms}
\label{sec:EA}
Consider the scalar diffusion process $X = (X_t: t \geq 0)$ satisfying the stochastic differential equation (SDE)
\[
dX_t = \mu(X_t) dt + \sigma(X_t) dB_t, \qquad X_0 = x,\ t \geq 0.
\]
with drift coefficient $\mu:\bbR\to\bbR$ and diffusion coefficient $\sigma:\bbR\to\bbR$. In all that follows we are interested in the law of the diffusion 
only up to some fixed, finite time $T$. We first apply the \emph{Lamperti transform} $X_t \mapsto \eta(X_t) =: Y_t$ defined via
\begin{equation}
\label{eq:lamperti}
\eta(X_t) = \int^{X_t}_\xi \frac{1}{\sigma(u)} du,
\end{equation}
for some fixed $\xi$ in the state space of $X$. The process $Y = (Y_t: t \geq 0)$ satisfies
\begin{equation}
\label{eq:ySDE}
dY_t = \alpha(Y_t) dt + dB_t, \qquad Y_0 = \eta(x) =: y,\ t \geq 0,
\end{equation}
where the new drift $\alpha:\bbR\to\bbR$ is given by
\begin{equation*}
\alpha(Y_t) = \frac{\mu(\eta^{-1}(Y_t))}{\sigma(\eta^{-1}(Y_t))} - \frac{1}{2}\sigma'(\eta^{-1}(Y_t)).
\end{equation*}
Equation \eref{eq:ySDE}, with unit diffusion coefficient, will be our central object of study, and we assume it admits a unique weak solution. The diffusion coefficient, $\sigma(\cdot)$, of a one-dimensional diffusion can always be reduced to unity in this manner, but we note that the infinitesimal covariance of a multidimensional diffusion may or may not be reducible to the identity matrix \citep[][gives an explicit way to check]{ait:2008}. 

The law of $Y$, to be denoted $\bbQ_y$, is absolutely continuous with respect to the law of Brownian motion commenced from $B_0 = y$, to be denoted $\bbW_y$. This motivates the latter's use for drawing candidates in a rejection sampling algorithm. The acceptance probability in such an algorithm would be proportional to the Radon-Nikod\'ym derivative of $\bbQ_y$ with respect to $\bbW_y$, which is given by Girsanov's formula \citep[e.g.][]{rev:yor:1999}:
\begin{equation}
\label{eq:girsanov}
\frac{d\bbQ_{y}}{d\bbW_{y}}(Y) = \exp\left\{\int_0^T \alpha(Y_t) dY_t - \frac{1}{2}\int_0^T \alpha^2(Y_t)dt\right\}.
\end{equation}
Equation \eref{eq:girsanov} makes it clear why such an algorithm is impossible: to compute the rejection probability one must evaluate integrals over the whole sample path. The \emph{exact} algorithms show how it is possible to make the accept/reject decision without this requirement, using only a finite amount of computation. I give here a brief overview; for further details the reader is referred to \citet{bes:rob:2005, bes:etal:2006:B, bes:etal:2008}.

To proceed we make three further assumptions.
\begin{itemize}
	\item[(A1)] The Radon-Nikod\'ym derivative of $\bbQ_{y}$ with respect to $\bbW_{y}$ exists and is given by \eref{eq:girsanov}.
	\item[(A2)] $\alpha \in C^1$, i.e.\ $\alpha$ is continuously differentiable.
	\item[(A3)] $\alpha^2 + \alpha'$ is bounded below.
\end{itemize}
Using (A2) and It\^o's lemma applied to $A(u) := \int^u_0 \alpha(z) dz$, 
we can rewrite \eref{eq:girsanov} as
\begin{equation*}
\frac{d\bbQ_{y}}{d\bbW_{y}}(Y) = \exp\left\{A(Y_T) - A(y) - \frac{1}{2}\int_0^T [\alpha^2(Y_t) + \alpha'(Y_t)]dt\right\}.
\end{equation*}
To simplify matters, rather than working directly with $\bbW_y$ we work with the probability measure $\bbZ_y$ defined via
\[
\frac{d\bbZ_y}{d\bbW_y}(B) \propto \exp\left\{A(B_T) - A(y)\right\}.
\]
The measure $\bbZ_y$ corresponds to \emph{biased} Brownian motion \citep{bes:rob:2005}; it is the law of $(B_t: t\in[0,T] \mid B_T \sim h)$, a Brownian motion whose endpoint has pre-specified density proportional to $h(u) := \exp\left\{-(u-y)^2/(2T) + A(u)\right\}$. Paths can be drawn according to $\bbZ_y$ first by simulating a random variable $B_T \sim h$ and then interpolating from $y$ to $B_T$ using the dynamics of a Brownian bridge. This is possible under the additional assumption
\begin{itemize}
\item[(A4)] $h(u) = \exp\left\{-(u-y)^2/(2T) + A(u)\right\}$ is integrable.
\end{itemize}
Notice that if our interest was to simulate not from \eref{eq:ySDE} but rather the sample paths of the corresponding \emph{bridge}, with $Y_T = z$, then we simply adjust this step by replacing $B_T \sim h$ with $B_T \sim \delta_z$.

Under (A4) we can use candidates from $\bbZ_y$ instead of $\bbW_y$; it is immediate that
\begin{equation}
\label{eq:dQdZ}
\frac{d\bbQ_{y}}{d\bbZ_{y}}(Y) 
\propto \exp\left\{-\int_0^T \phi(Y_t) dt\right\} \leq 1,\qquad \bbZ_y\text{-a.s.,}
\end{equation}
where
\[
\phi(u) := \frac{1}{2}[\alpha^2(u) + \alpha'(u)] - \inf_{z\in \bbR} \frac{1}{2}[\alpha^2(z) + \alpha'(z)]
\]
is well-defined by (A2, A3). 
As a consequence of \eref{eq:dQdZ}, if we can propose a candidate path $Y = (Y_t: t \in [0,T])$ from $\bbZ_{y}$ and accept it with probability $\exp\left\{-\int_0^T \phi(Y_t) dt\right\}$, then the accepted paths are distributed according to $\bbQ_{y}$. It remains to construct an event, call it $\Gamma$, occurring with the required probability. Inspection of the form of \eref{eq:dQdZ} suggests one way to define $\Gamma$: as the event that there are no points in the realization of a Poisson process occurring at unit rate in the area under the graph of $t \mapsto \phi(Y_t)$. A practical way to achieve this is to thin a Poisson process occurring in a larger rectangle containing this graph. More precisely, suppose there exists a random variable $\Upsilon$ and a positive function $r$ such that $r(\Upsilon) < \infty$ a.s.\ and $\sup_{t\in[0,T]}\phi(Y_t) \leq r(\Upsilon)$ a.s. Let $\Phi_{r(\gU)} = \{(\chi_j,\psi_j) : j\geq 1\}$ denote the points of a unit-rate Poisson process on $[0,T]\times [0,r(\Upsilon)]$, and let $\text{epi}[\phi(Y)]$ denote the epigraph of $t \mapsto \phi(Y_t)$:
\[
\text{epi}[\phi(Y)] := \{(s,u) \in [0,T] \times \bbR_+ : \phi(Y_s) \leq u\}.
\]
Then the event $\Gamma := \left\{\Phi_{r(\gU)} \subseteq \text{epi}[\phi(Y)]\right\}$ occurs with probability $\exp\left\{-\int_0^T \phi(Y_t) dt\right\}$, as required \citep{bes:etal:2008}.

The point of first finding an a.s.\ upper bound $r(\Upsilon)$ on $\phi(Y_t)$ is to ensure the Poisson process has finite total rate, and hence that $\Gamma$ can be determined with a finite amount of computation. Assuming such a bound can be found, an exact algorithm therefore proceeds as follows:\\

\begin{minipage}[h]{\textwidth}
{\bf Exact algorithm (EA)}

\begin{tabular}{lp{0.8\textwidth}}
\hline
1. & Simulate $Y_T \sim h$.\\
2. & Simulate $\gU$ conditionally on $Y_T$.\\
3. & Simulate $\Phi_{r(\gU)}$.\\
4. & Simulate $\{Y_{\chi_i} : 1 \leq i \leq |\Phi_{r(\gU)}\}|$ from $\bbZ_y \mid Y_T, \Upsilon$.\\
5. & If $\Gamma$ has occurred output $\cS := \{(0,y),(T,Y_T)\} \cup \{Y_{\chi_i} : 1 \leq i \leq |\Phi_{r(\gU)}|\}$, otherwise return to 1.\\
\hline
\end{tabular}
\end{minipage}
\\
\\

The output of EA is a set $\cS$ of random skeleton points. Once a skeleton has been accepted, further points can be filled in by sampling from Brownian bridges between each skeleton point (conditional on $\gU$). No further reference to the target distribution $\bbQ_y$ is necessary. 

It remains to find a suitable bound, $r(\gU)$, which must be defined in such a way that steps 2 and 4 can still be carried out. 
To construct the required bound, further regularity conditions on $\phi(Y)$ are required. These were weakened in a series of papers and lead to successively more general, though algorithmically more complicated, versions of the exact algorithm. Restrictions on $\phi$ are as follows.

\subsection{EA1 \citep{bes:rob:2005}}
Assume that 
\begin{itemize}
	\item[(*)] The function $\phi(u)$ is bounded above.
\end{itemize}
Then the non-random choice $r(\gU) := \sup_{u \in \bbR} \phi(u)$ suffices. Step 2 of EA can be omitted, and step 4 simplifies to simulating from the finite-dimensional distributions of a Brownian bridge.

\subsection{EA2 \citep{bes:etal:2006:B}}
\label{sec:EA2}
Assume that
\begin{itemize}
	\item[(**)] Either $\lim\sup_{u \to \infty} (\alpha^2 + \alpha')(u) < \infty$ or $\lim\sup_{u \to -\infty} (\alpha^2 + \alpha')(u) < \infty$.
\end{itemize}
Without loss of generality suppose the former: then $\phi$ is bounded on $[u,\infty)$ for any $u\in\bbR$. Hence, for $\gU = m_T := \inf_{t\in[0,T]}Y_t$ the choice $r(\gU) := \sup_{u\in[\gU,\infty)}\phi(u)$ provides the required bound. Step 2 then corresponds to the simulation of the minimum of a Brownian bridge, and step 4 corresponds to simulation from the finite-dimensional distributions of a Brownian bridge conditional on its minimum. These distributions are known in closed-form. In practice the minimum is simulated together with the (a.s.\ unique) time it is attained.
 
\subsection{EA3 \citep{bes:etal:2008}}
\label{sec:EA3}
This algorithm makes no further assumptions on $\phi(u)$. Intuitively, to relax (**) we would like to simulate both the minimum, $m_T$, and maximum, $M_T := \sup_{t\in[0,T]} Y_t$, of the candidate path, and then use 
$r(\gU) := \sup_{u\in[m_T,M_T]} \phi(u)$. This is not quite within reach, but a closely related idea \emph{is} feasible. We specify a partition of the state space of $Y_t$ and simulate the member of the partition into which the more `extreme' of $m_T$ and $M_T$ falls. More precisely, let $\{a_i\}_{i\geq 1}$ and $\{b_i\}_{i \geq 1}$ be two increasing sequences of positive real numbers with $a_0 = b_0 = 0$. Given $Y_0 = y$ and $Y_T = z$ with $\yb := y \wedge z$, $\zb := y \vee z$, define
\begin{equation}
\begin{split}
\label{eq:UI,LI}
U_i &:= \left\{M_T \in [\zb + b_{i-1}, \zb + b_i)\right\} \cap \left\{m_T > \yb - a_i\right\},\\
L_i &:= \left\{m_T \in (\yb - a_i, \yb - a_{i-1}]\right\} \cap \left\{M_T < \zb + b_i\right\},
\end{split}
\end{equation}
and $D_i := U_i \cup L_i$. The sets $\{D_i: i\geq 1\}$ partition the state space of $Y_t$ into layers. 
We may construct a discrete random variable, call it $I$, representing the choice of layer, so that $D_i = \{I = i\}$. Setting $\gU = I$, a suitable bound is then given by $r(\gU) := \sup\{\phi(u): u \in (\yb - a_i, \zb + b_i)\}$.


Step 2 of EA then requires simulation of the layer $I$ of $Y$, and Step 4 requires the simulation of points from $Y$ given its layer, each of which can be achieved; see \citet{bes:etal:2008} for details. 

The exact algorithms EA1--3 are extremely appealing because they do not suffer from discretization error, and are simple to implement. They enable candidate (biased) Brownian motions to be rejected with exactly the right probability so that accepted paths are distributed according to the target law. However, one disadvantage is that 
they offer limited control over the rejection probability. To the extent that Brownian motion may poorly resemble the target diffusion, the rejection rate could be high. As discussed in the Introduction, a diffusion with a finite entrance boundary is one example of this. In \sref{sec:Bessel-EA} I address this issue by replacing Brownian motion in the exact algorithm with the \emph{Bessel} process. First, let us briefly record some facts about the Bessel process for later use, which may be found for example in \citet{rev:yor:1999}.

\section{The Bessel process}
\label{sec:Bessel}
For any real $\gd \geq 0$, the Bessel process of dimension $\gd$ (equivalently, of order $\nu = (\gd - 2)/2$), commenced from $y$, is the diffusion $(\B_t : t \geq 0)$ with $\B_t$ taking values in $[0,\infty)$, $\B_0 = y$, with generator
\[
\cL = \frac{\gd - 1}{2x}\frac{d}{dx} + \frac{1}{2}\frac{d^2}{dx^2}.
\]
In the terminology of Feller \citep[\S 15.6]{kar:tay:1981}, the boundary $0$ is classified as
\[
\begin{cases}
\text{entrance (not exit)} & \text{if }\gd \geq 2,\\
\text{regular (entrance and exit)} & \text{if }0 < \gd < 2,\\
\text{exit (not entrance)} & \text{if }\gd = 0.
\end{cases}
\]
It is necessary to further specify the behaviour of a regular boundary; for the Bessel process with $0 < \gd < 2$ the boundary at $0$ is instantaneously reflecting. The boundary at $\infty$ is natural for all $\gd \geq 0$.

We will denote the probability measure on $\B = (\B_t : t \geq 0)$ by $\bbB^\gd_y$; we will also make use of the \emph{square} of a Bessel process, which is also a diffusion process and whose measure we denote $\bbB\bbQ^\gd_{y^2}$. 
It is well known that for $\gd \in \bbZ_+$ the Bessel process is the radial part of a Brownian motion in $\bbR^\gd$. It inherits a number of nice properties from Brownian motion, for \emph{all} $\gd \geq 0$, including Brownian scaling, bridge constructions, time reversal, and a known transition density: 
for $z > 0$ and $\gd > 0$ it is given by 
\begin{equation}
\label{eq:transitiondensity}
p^\gd_t(y,z) =
\begin{cases}
\ds\frac{1}{t}\frac{z^{\nu + 1}}{y^\nu}\exp\left(-\frac{y^2 + z^2}{2t}\right)I_\nu\left(\frac{yz}{t}\right), & y > 0,\\
\ds\frac{1}{2^\nu t^{\nu+1}\Gamma(\nu + 1)}z^{2\nu + 1}\exp\left(-\frac{z^2}{2t}\right), & y = 0.

\end{cases}
\end{equation}

An analogous set of results holds for $\bbB\bbQ^\gd_{y^2}$. 


\section{An exact algorithm using the Bessel process}
\label{sec:Bessel-EA}

I now develop an exact algorithm using the Bessel process to construct candidate paths, where our interest is in a target diffusion with a finite entrance boundary. By translating and reflecting $Y$ if necessary, there is no loss of generality in assuming it to be a lower boundary at $0$; we will obviously require $\bbQ_y \ll \bbB^\gd_y$, and in particular that its state space is 
$(0,\infty)$. In this section we assume for simplicity that $\infty$ is a natural boundary and that $y > 0$. 

Recall that the drift of our target, $\bbQ_y$, is $\ga(x)$ and the drift of the Bessel process, $\bbB^\gd_y$, is $\gb(x) := (\gd-1)/(2x)$. We fix a $\gd \geq 2$ to specify our candidate process. Following \citet{sch:etal:2013},
\begin{align}
\frac{d\bbQ_{y}}{d\bbB^\gd_{y}}(Y) &= \frac{d\bbQ_{y}}{d\bbW_{y}}(Y)\Big/\frac{d\bbB^\gd_{y}}{d\bbW_y}(Y) \label{eq:RDchain}\\
&= \exp\left\{\int_0^T [\alpha(Y_t)-\beta(Y_t)] dY_t - \frac{1}{2}\int_0^T [\alpha^2(Y_t) - \beta^2(Y_t)] dt\right\}.
\label{eq:girsanov2}
\end{align}
We replace (A1--4) with the following assumptions:
\begin{itemize}
	\item[(B1)] The Radon-Nikod\'ym derivative of $\bbQ_{y}$ with respect to $\bbB^\gd_{y}$ exists and is given by \eref{eq:girsanov2}.
	\item[(B2)] $\alpha \in C^1$ on $(0,\infty)$.
	\item[(B3)] $\alpha^2 - \beta^2 + \alpha' - \beta'$ is bounded below on $(0,\infty)$.
	\item[(B4)] $\widetilde{h}(u) = p_t^\gd(y,u)\exp\left\{\widetilde{A}(u)\right\}$ is integrable, where $\widetilde{A}(u) := \int_{u_0}^u [\ga-\gb](z) dz$ for some $u_0 > 0$.
\end{itemize}
Assumption (B1) is easily met for a diffusion with an entrance boundary and up to a finite time $T$ because we are then almost surely on the set $\{T_0 > T\}$, where $T_c := \inf\{t \geq 0: Y_t = c\}$ \citep[see][for discussion of the case $T = \infty$]{pit:yor:1981}. Equation \eref{eq:girsanov2} shows that we should try to choose $\gd$ to minimize this exponentiand, which, for a target diffusion whose drift has a singularity at 0, is typically achieved by selecting the $\gd$ to eliminate this singularity. 

\citet{bes:etal:2008} remark that (A4) is quite weak and is satisfied (for small $t$) by a linear growth bound on $\ga$. The same may be said of (B4) since the transition densities of both Brownian motion and the Bessel process contain a controlling term $\exp\{-u^2/(2t)\}$ as $u \to \infty$.

As in the Brownian case, we introduce a \emph{biased} Bessel process, denoted $\bbB\bbZ_y^\gd$ and defined via
\begin{equation}
\label{eq:biasedBessel}
\frac{d\bbB\bbZ^\gd_y}{d\bbB^\gd_y}(\B) \propto \exp\left\{\widetilde{A}(\B_T) - \widetilde{A}(y)\right\}.
\end{equation}
To simulate using $\bbB\bbZ_y^\gd$, we will first draw $\B_T \sim \widetilde{h}$ and then simulate the rest of the path using the dynamics of the corresponding bridge; this time we must simulate from a \emph{Bessel} bridge.

The main result of this paper, which enables exact simulation from $\bbQ_y$ using a (biased) Bessel process, is as follows.
\begin{theorem}
$\bbQ_y$ is the marginal distribution of $Y$ when $(Y, \Phi) \sim \bbB\bbZ_y^\gd \otimes \bbL \mid \widetilde{\Gamma}$, where $\bbL$ is the law of a Poisson point process of unit rate on $[0,T] \times [0,\infty)$, 
\[
\widetilde{\Gamma} := \left\{\Phi \subseteq \mathrm{epi}\left[\widetilde{\phi}(Y)\right]\right\},
\]
and
\[
\widetilde{\phi}(u) := \frac{1}{2}[\alpha^2(u) - \gb^2(u) + \alpha'(u) - \gb'(u)] - \inf_{z\in (0,\infty)} \frac{1}{2}[\alpha^2(z) - \gb^2(z) + \alpha'(z) - \gb'(z)].
\]
\end{theorem}
\begin{proof}
Note that $\widetilde{\phi}$ is well-defined and continuous by (B2, B3). Using \eref{eq:girsanov2} and It\^o's lemma applied to $\widetilde{A}(u)$, we obtain
\begin{equation*}
\frac{d\bbQ_{y}}{d\bbB^\gd_{y}}(Y) = \exp\left\{\widetilde{A}(Y_T) - \widetilde{A}(y) - \frac{1}{2}\int_0^T [\alpha^2-\gb^2 + \alpha' - \gb'](Y_t)dt\right\}.
\end{equation*}
Hence, using \eref{eq:biasedBessel},
\[
\frac{d\bbQ_y}{d\bbB\bbZ^\gd_y}(Y) = \frac{d\bbQ_y}{d\bbB^\gd_y}(Y)\Big/ \frac{d\bbB\bbZ^\gd_y}{d\bbB^\gd_y}(Y) \propto \exp\left\{-\int_0^T \widetilde{\phi}(Y_t) dt\right\} = \bbL(\widetilde{\Gamma}\mid Y),
\]
so, for a measurable set $A$ we have
\[
(\bbB\bbZ_y^\gd \otimes \bbL)(A\mid \widetilde{\Gamma}) = \int_A \frac{\bbL(\widetilde{\Gamma}\mid Y) d\bbB\bbZ_y^\gd (Y)}{(\bbB\bbZ_y^\gd \otimes \bbL)(\widetilde{\Gamma})} \propto \int_A \frac{d\bbQ_y}{d\bbB\bbZ^\gd_y}(Y) d\bbB\bbZ_y^\gd(Y) = \bbQ_y(A),
\]
as required.
\end{proof}
The key observation that
\[
\frac{d\bbQ_y}{d\bbB\bbZ^\gd_y}(Y) \propto \exp\left\{-\int_0^T \widetilde{\phi}(Y_t) dt\right\} \leq 1,\qquad \bbB\bbZ^\gd_y\text{-a.s.},
\]
provides the rejection probability underlying our algorithm. Practically, we need to simulate whether or not $\widetilde{\Gamma}$ has occurred. By the same argument as in the Brownian case, this can be simulated by checking that the points of a unit rate Poisson process on $[0,T] \times [0, r(\gU)]$ fall within the epigraph of $t \mapsto \widetilde{\phi}(Y_t)$, where $r(\gU)$ is a.s.\ finite and an a.s.\ upper bound on $\sup_{t\in[0,T]}\widetilde{\phi}(Y_t)$, and $\gU$ is a random variable which may depend on $Y$. Hence, an exact algorithm proceeds as follows:\\

\begin{minipage}[h]{\textwidth}
{\bf Exact algorithm---Bessel version (Bessel-EA)}

\begin{tabular}{lp{0.8\textwidth}}
\hline
1. & Simulate $Y_T \sim \widetilde{h}$.\\
2. & Simulate $\gU$ conditionally on $Y_T$.\\
3. & Simulate $\Phi_{r(\gU)}$.\\
4. & Simulate $\{Y_{\chi_i} : 1 \leq i \leq |\Phi_{r(\gU)}\}|$ from $\bbB^\gd_y \mid Y_T, \Upsilon$.\\
5. & If $\Gamma$ has occurred output $\cS := \{(0,y),(T,Y_T)\} \cup \{Y_{\chi_i} : 1 \leq i \leq |\Phi_{r(\gU)}|\}$, otherwise return to 1.\\
\hline
\end{tabular}
\end{minipage}
\\
\\

Given an accepted skeleton, further points can be filled in by simulation from Bessel bridges between each skeleton point (conditional on $\gU$). 
No further reference to $\bbQ_y$ is necessary.

It seems as though we have just replaced one proposal measure for another. The benefit of this operation is apparent when we notice that, for certain diffusions, the hypograph of $t \mapsto \widetilde{\phi}(Y_t)$ may be dramatically smaller than the hypograph of $t \mapsto \phi(Y_t)$, manifesting itself via a much lower rejection rate. For example, if our target drift is of the form
\begin{equation*}
\alpha(x) = \frac{\gd - 1}{2x} + o(1),
\end{equation*}
as $x \to 0$, then $\phi(x) \to\infty$ whereas $\widetilde{\phi}(x)$ remains finite. 
Thus, we should expect the rejection rate to be improved most when the effect of the boundary is strong and/or the target process spends a great deal of time near the boundary, for example if one end of the bridge is 
close to $0$. These observations are verified in an example application later. First, we must specify an appropriate choice of $\gU$ and $r$, and how to simulate from $\bbB^\gd_y \mid Y_T, \Upsilon$, 
which requires further assumptions on $\widetilde{\phi}$. 
For simplicity we will focus on the Bessel analogue of EA1, to be denoted Bessel-EA1, for which we assume
\begin{itemize}
	\item[(B*)] The function $\widetilde{\phi}(u)$ is bounded above.
\end{itemize}
Under this assumption we are entitled to choose the non-random $r(\gU) = \sup_{u \in (0,\infty)} \widetilde{\phi}(u)$ and to omit Step 2 of Bessel-EA. It remains to simulate skeleton points from the law of a Bessel bridge. 
I now detail how this can be achieved.

When $\gd \in \bbZ_+$ it is well known that the Bessel process is the radial part of a Brownian motion in $\bbR^\gd$. It is possible to use this observation to simulate from a Bessel bridge by transforming an underlying Brownian bridge in $\bbR^\gd$ \citep[see for example][]
{sch:etal:2013}. However, it is in fact possible to simulate exactly from the Bessel bridge for \emph{any} real $\gd \geq 0$, as follows. We first need a definition.

A random variable $W$ on $\bbN = \{0,1,\ldots\}$ is said to be \emph{Bessel$(\nu,x)$-distributed} when
\[
b_{\nu,x}(n) := \bbP(W = n) = \frac{(x/2)^{2n + \nu}}{n!\Gamma(\nu + n + 1)I_\nu(x)},\qquad x > 0,\ n\in \bbN.
\]
This distribution is constructed by normalizing the coefficients of
\[
I_\nu(x) = \sum_{n=0}^\infty \frac{(x/2)^{2n + \nu}}{n!\Gamma(\nu+n+1)}
\]
to sum to $1$. We define $b_{\nu,0}(n)$ as the continuous limit as $x\downarrow 0$; then $W = 0$ a.s. Because the Bessel distribution is discrete, a realization of $W$ can be achieved easily by the usual method of simulating a Uniform$[0,1]$ random variable and inverting the cumulative distribution function of $W$.

\begin{proposition}
\label{prop:mak:gle:2010}
Suppose $W \sim\text{Bessel}\left(\nu,\frac{\sqrt{yz}}{T}\right)$ and, independently,
\[
V \sim \text{Poisson}\left(\frac{1}{2}\left[y\frac{T-t}{tT} + z\frac{t}{T(T-t)}\right]\right),
\]
where $0 < t < T$ and $\gd = 2\nu + 2 > 0$. Then marginally under $\bbB\bbQ^\gd_{y\to z,T}$,
\[
\B^2_t \sim \text{Gamma}\left(V + 2W + \nu + 1, \frac{T}{2t(T-t)}\right),
\]
where $\text{Gamma}(S,R)$ is a Gamma distribution with shape $S$ and rate $R$.
\end{proposition}
\begin{proof}
See \citet{mak:gle:2010}.
\end{proof}
Using the Markov property this result can be applied repeatedly to simulate a set of skeleton points $\{(\chi_i,\B_{\chi_i}): 1 \leq i \leq |\Phi_{r(\gU)}|\}$. Furthermore, the rest of the path can be filled in by the following theorem.

\begin{theorem}
Let $\bbB\bbQ^\gd_{(s_1,y_1)\to (s_2,y_2)}$ denote the law of a squared Bessel bridge commencing from $y_1$ at time $s_1$ and ending at $y_2$ at time $s_2$. If $\{s_i: 0 \leq i \leq |\Phi_{r(\gU)}| + 1\}$ is an increasing ordering of the times of the points in $\cS$ output from Bessel-EA1 under (B1--4) and (B*), then the rest of the path, $\B = \{\B_t: t \in [0,T]\}$ is distributed as
\[
\B^2 \mid \cS, \widetilde{\Gamma} \sim \bigotimes_{i=1}^{|\Phi_{r(\gU)}|+1} \bbB\bbQ^\gd_{(s_{i-1},\B^2_{s_{i-1}})\to (s_i,\B^2_{s_i})},
\]
where $\B^2 = \{\B^2_t: t \in [0,T]\}$.
\end{theorem}
\begin{proof}
This is analogous to Theorem 2 of \citet{bes:etal:2008}, and follows from the observation that $\ggg \mid \cS, \widetilde{\Gamma} \stackrel{d}{=} \ggg \mid \cS$.
\end{proof}

\section{Applications}
\label{sec:applications}
\subsection{Conditioned diffusions}
\label{sec:conditioned}
A rich source of processes to which the above theory applies can be constructed as follows. Suppose we have a diffusion satisfying \eref{eq:ySDE}, (A1--4), and (*), i.e.\ it can be simulated by EA1. Then the diffusion $Y^*$ obtained by conditioning this process on $\{T_b < T_0\}$, i.e.\ on reaching some high level $b$ before hitting $0$, has new drift and diffusion coefficients given by
\begin{equation}
\label{eq:conditioned}
	\ga^*(Y_t) = \ga(Y_t) + \frac{S'(Y_t)}{S(Y_t) - S(0)},\qquad
	(\gs^*)^2(Y_t) = 1,
\end{equation}
where
\[
S(y) = \int^y \exp\left[-2\int^x \ga(z) dz\right] dx
\]
is the scale function of the diffusion \citep[\S 15.9]{kar:tay:1981}. Denote the law of this diffusion by $\bbQ^*_y$. For such a process, we have the following result.
\begin{theorem}
\label{thm:conditioned}
	With the choice $\gd = 3$, the process $Y^*$ satisfies (B1--4) and (B*).
\end{theorem}
\begin{proof}
	(B1) follows immediately from (A1) and the observation that $\bbQ^*_y \ll \bbQ_y$; (B2) follows from (A2) and continuity of $S$ and $S'$; and (B4) follows from (A4). It remains to check that the function $(\ga^*)^2 - \gb^2 + (\ga^*)' - \gb'$ is bounded on $(0,\infty)$. A direct calculation shows that
\[
[(\ga^*)^2 - \gb^2 + (\ga^*)' - \gb'](u) = \ga^2(u) + \ga'(u) + \frac{(\gd - 3)(\gd - 1)}{4u^2},
\]
which is bounded by (A3) and (*) when $\gd = 3$.
\end{proof}
Thus, using a $\bbB_y^3$ process builds the conditioning into every candidate path. An alternative strategy using Brownian paths would be first to simulate the minimum of a $\bbW_y$ process conditioned on its minimum being positive, and then to simulate the requisite bridges either side of this minimum. This strategy is closely related to that of \citet{bes:etal:2006:B}, and since these bridges are also in fact Bessel bridges of dimension 3, the two strategies are expected to have very similar performance.

\subsection{A population growth model}
Consider the general population growth model $(X_t: t \geq 0)$, $X_t \in [0,\infty)$, with generator
\begin{equation}
\label{eq:growthgen}
\cL = \gk x \frac{d}{dx} + \frac{1}{2}(\gt x + \go x^2)\frac{d^2}{dx^2},
\end{equation}
where $-\infty < \gk < \infty$, $\gt \geq 0$, $\go \geq 0$, and $\gt + \go > 0$ \citep[p378]{kar:tay:1981}. This model includes as special cases the commonly used growth models of geometric Brownian motion ($\gt = 0$) and the squared Bessel process ($\go = 0$). In the latter case our algorithm would have optimal performance, with rejection rate 0.

Ascertainment bias directs interest to the growth trajectories of populations observed to have grown successfully; such populations follow the diffusion \eref{eq:growthgen} conditioned as in \sref{sec:conditioned}. Here, \thmref{thm:conditioned} does not apply because the unconditioned diffusion does not satisfy (A2); however, we can still perform exact simulation using a Bessel process, as follows. We assume $\go \neq 2\gk$ and $\gk > 0$ (other cases are similar but omitted for brevity).

Applying the Lamperti transform \eref{eq:lamperti} to \eref{eq:growthgen} and conditioning via \eref{eq:conditioned} yields, after extensive but routine calculations, a diffusion on $(\ln \gt/\sqrt{w},\infty)$ with drift and diffusion coefficients
\begin{align}
\ga(y) ={} & \frac{\gk}{\sqrt{\go}}\tanh\left[\frac{\sqrt{\go}y - \ln \gt}{2}\right] - \frac{\sqrt{\go}}{2}\coth\left[\sqrt{\go}y - \ln \gt\right]\notag\\
& {}+ \frac{\go - 2\gk}{\sqrt{w}}\frac{\tanh\left[\frac{\sqrt{\go}y - \ln\gt}{2}\right]}{1 - \cosh^{\frac{4\gk}{\go} - 2}\left[\frac{\sqrt{\go}y - \ln\gt}{2}\right]},\label{eq:growthdrift}\\
\sigma^2(y) ={} & 1,\notag
\end{align}
This diffusion has an entrance boundary at $\ln\gt/\sqrt{w}$, so we introduce $z = y - \ln\gt/\sqrt{w}$. A Taylor expansion of \eref{eq:growthdrift} about $z = 0$ shows that
\begin{equation}
\label{eq:growthdriftlimit}
\ga(z) = \frac{3}{2z} + O(z)
\end{equation}
as $z \to 0$, so that a suitable candidate diffusion is a Bessel process of dimension $\gd = 4$. It is straightforward to verify that the assumptions of Bessel-EA1 apply. In particular,
\[
[\ga^2 + -\gb^2 + \ga' - \gb'](z) = \frac{(\go-2\gk)^2}{4\go} + \frac{3\go + 8\gk(\cosh(\sqrt{\go}z) - 1)}{4\sinh^2(\sqrt{\go}z)}- \frac{\gk^2}{\go\cosh^2(\sqrt{\go}z/2)} - \frac{3}{4z^2},
\]
which we can show to be bounded on $(0,\infty)$ as follows. Note that, since $\sinh x \geq x$ and $\cosh x \geq 1$ when $x \geq 0$,
\[
\frac{3\go + 8\gk(\cosh(\sqrt{\go}z) - 1)}{4\sinh^2(\sqrt{\go}z)} - \frac{3}{4z^2} \leq 2\gk\frac{\cosh(\sqrt{\go}z) - 1}{\sinh^2(\sqrt{\go}z)} = 2\gk\frac{\cosh(\sqrt{\go}z) - 1}{\cosh^2(\sqrt{\go}z) - 1} \leq 2\gk.
\]
Hence,
\[
[\ga^2 + -\gb^2 + \ga' - \gb'](z) \leq \frac{(\go-2\gk)^2}{4\go} + 2\gk.
\]
Similarly, writing
\[
f(z) := \frac{3}{4}\left[\frac{\go}{\sinh^2(\sqrt{\go}z)} - \frac{1}{z^2}\right],
\]
we find
\[
f'(z) = \frac{3}{2}\left[\frac{1}{z^3} - \frac{\go^{3/2}\cosh(\sqrt{\go}z)}{\sinh^3(\sqrt{\go}z)}\right] \geq 0,
\]
where the inequality follows from $\sinh(x)/x \geq \cosh^{1/3}(x)$ \citep{bul:1998}. Hence, $f(z) \geq f(0+) = -\go/4$, and so
\[
[\ga^2 + -\gb^2 + \ga' - \gb'](z) \geq \frac{(\go-2\gk)^2}{4\go} - \frac{w}{4} + 2\gk\frac{\cosh(\sqrt{\go}z) - 1}{\sinh^2(\sqrt{\go}z)} \geq -\gk.
\]
I applied Bessel-EA1 to simulate sample paths of this diffusion, from $Y_0 = y$ to $Y_T = z$. For comparison I also implemented a slight modification of EA2 as discussed in \sref{sec:conditioned}, which uses Brownian motion as its candidate paths: In order to bound $\phi(Y_t)$, the EA2 algorithm first simulates the minimum of a Brownian bridge (and to respect the lower boundary I further condition it to be positive) and then fills in skeleton points conditioned on this minimum (see \sref{sec:EA2}). To compare the two algorithms, I used two measurements of running time: the total number of random variates generated, and the total running time in seconds. Each algorithm requires 2$|\Phi_{r(\gU)}|+1$ realizations of random variables to simulate $\Phi_{r(\gU)}$. Additionally, each skeleton point requires the simulation of a coordinate from a Bessel(4)- (Bessel-EA1) or Bessel(3)- (EA2) bridge, requiring four or three random variables respectively. EA2 suffers an additional one-off cost of 6 random variables in order to simulate $\gU$. Thus, EA2 requires up to $5|\Phi_{r(\gU)}| +7$ realizations per candidate path (whether accepted or not), while Bessel-EA1 requires up to $6|\Phi_{r(\gU)}| + 1$. The relative performance of the two algorithms therefore depends on three factors: (i) $|\Phi_{r(\gU)}|$, which in turn depends on $r(\gU)$, the size of the bound on $\phi$ or $\widetilde{\phi}$; (ii) the rejection rate, which depends on the distribution of $(\phi(Y_t): t\in [0,T])$ [or $(\widetilde{\phi}(Y_t): t\in [0,T])$]; and (iii) the cost of generating each variate. For both algorithms all random variables come from well known distributional families and are easy to simulate. 

\begin{table}[p]
\caption{\label{tab:KTgrowth}Performance comparison of (a) Bessel-EA1 versus (b) EA2, applied to a bridge of the diffusion \eref{eq:growthdrift} from $Y_0 = y$ to $Y_{0.15} = 1$. Each quoted value (except the total running time) is per accepted sample path, averaged across 10,000 accepted paths. In all simulations the remaining parameters are $\go = 3$, $\tau = 1$. Entries marked `--' could not be completed.}
\begin{center}
\begin{tabular}{r@{.}l r@{.}l r@{.}l r@{.}l r@{.}l r@{.}l r}
\multicolumn{13}{c}{(a) Bessel-EA1}\\
\hline
\multicolumn{2}{c}{} & \multicolumn{2}{c}{} & \multicolumn{2}{c}{} 
 & \multicolumn{2}{r}{Poisson} & \multicolumn{2}{r}{Skeleton} & \multicolumn{2}{c}{Random} & Total\phantom{ (s)} \\
\multicolumn{2}{c}{$\gk$} & \multicolumn{2}{c}{$y$} & \multicolumn{2}{c}{Attempts} 
 & \multicolumn{2}{r}{points} & \multicolumn{2}{r}{points} &  \multicolumn{2}{r}{variables} & Time (s)\\
\hline
1&0  & 10&0 &              1&1 &             0&2 &             0&2 &             1&9 &               0 \\
1&0  & 1&0 & 1&0 &             0&2 &             0&2 &             1&9 &               0 \\
1&0  & 0&5 &    1&0 &             0&2 &             0&2 &             1&9 &               0 \\
1&0  & 0&25 & 1&0 &             0&2 &             0&2 &             2&0 &               0 \\
1&0  & 0&15 &1&0 &             0&2 &             0&2 &             2&0 &               1 \\
1&0  & 0&1 & 1&1 &         0&2 &               0&2 &          2&0 &               1 \\
1&0  & 0&025 &1&0 &           0&2 &          0&2 &          2&0 &               0 \\

10&0  & 10&0 &             5&2 &            14&1 &             6&8 &            56&4 &               1 \\
10&0  & 1&0 &             3&0 &             7&9 &             4&9 &            36&4 &               1 \\
10&0  & 0&5 &             2&4 &             6&6 &             4&5 &            32&3 &               1 \\
10&0  & 0&25 &             2&3 &             6&1 &             4&4 &            30&8 &               1 \\
10&0  & 0&15 &             2&2 &             6&0 &             4&3 &            30&3 &               0 \\
10&0  & 0&1 &          2&2 &                         5&9 &                        4&4 &         30&4 &               0 \\
10&0  & 0&025 &          2&1 &                         5&8 &                        4&3 &         29&6 &               1 \\

\hline
\multicolumn{13}{c}{}\\
\multicolumn{13}{c}{(b) EA2}\\

\hline
\multicolumn{2}{c}{} & \multicolumn{2}{c}{} & \multicolumn{2}{c}{} 
 & \multicolumn{2}{r}{Poisson} & \multicolumn{2}{r}{Skeleton} & \multicolumn{2}{c}{Random}  & Total\phantom{ (s)} \\
\multicolumn{2}{c}{$\gk$} & \multicolumn{2}{c}{$y$} & \multicolumn{2}{c}{Attempts} 
 & \multicolumn{2}{r}{points} & \multicolumn{2}{r}{points} &  \multicolumn{2}{r}{variables} & Time (s)\\
\hline
1&0  & 10&0 &  1&0 &             0&1 &             0&1 &             7&3 &               0 \\
1&0  & 1&0 &  1&1 &             0&1 &             0&1 &             7&4 &               0 \\
1&0  & 0&5 & 1&1 &             0&3 &             0&3 &             8&5 &               0 \\
1&0  & 0&25 & 1&2 &          1288&6 &           420&6 &          3846&1 &               6 \\
1&0  & 0&15 & 1&4 &          7531&1 &           617&4 &         16921&4 &              16 \\
1&0  & 0&1 &    \multicolumn{2}{c}{--}  &                     \multicolumn{2}{c}{--}  &                        \multicolumn{2}{c}{--}  &            \multicolumn{2}{c}{--}  &              --  \\
1&0  & 0&025 &    \multicolumn{2}{c}{--}  &                     \multicolumn{2}{c}{--}  &                        \multicolumn{2}{c}{--}  &            \multicolumn{2}{c}{--}  &              --  \\

10&0  & 10&0 &             5&0 &             9&8 &             4&8 &            40&9 &               0 \\
10&0  & 1&0 &             2&9 &             5&9 &             3&6 &            29&8 &               0 \\
10&0  & 0&5 &             2&6 &             6&1 &             4&0 &            31&0 &               0 \\
10&0  & 0&25 &             2&6 &            81&4 &            10&7 &           201&9 &               0 \\
10&0  & 0&15 &             2&9 &         23052&1 &          1981&9 &         52056&9 &              52 \\
10&0  & 0&1 & \multicolumn{2}{c}{--}  &                     \multicolumn{2}{c}{--}  &                        \multicolumn{2}{c}{--}  &            \multicolumn{2}{c}{--}  &              --  \\
10&0  & 0&025 & \multicolumn{2}{c}{--}  &                     \multicolumn{2}{c}{--}  &                        \multicolumn{2}{c}{--}  &            \multicolumn{2}{c}{--}  &              --  \\

\hline
\end{tabular}
\end{center}
\end{table}

I simulated $10,000$ bridges for each of various combinations of model parameters $\gk, \go$ and initial position $y$. Reported in \tref{tab:KTgrowth} are the mean number, per accepted path, of: initiated attempts, Poisson points ($|\Phi_{r(\gU)}|$), skeleton points, and total number of r.v.s. The total running time (in seconds) is also shown. (The number of skeleton points can differ from the number of Poisson points because a candidate can be abandoned as soon as a point is found to fall outside the relevant epigraph.) As is clear from the table, Bessel-EA1 is only moderately sensitive to the parameters and proximity to the boundary, with around 1--5 attempts needed per accepted path. Where paths are unlikely to approach the boundary, such as when both endpoints are sufficiently far from 0, EA2 exhibits a similar performance both in the number of attempts per accepted path and the total running time. However, EA2 is highly sensitive to the proximity of $y$ to $0$, with the number of r.v.s required increasing rapidly as $y$ approaches $0$. This has serious implications for the memory requirements of the algorithm. While each execution of Bessel-EA1 required roughly 1MB, the requirements of EA2 quickly exceeded 6GB for $y \leq 0.1$, rendering the algorithm impracticable (at least without further optimizations). This is unfortunate because it precludes simulation of the trajectories of newly founded (i.e.\ initially small) populations. For both algorithms, increasing $\gk$ causes a moderate drop in performance because of a consequent increase in $r(\gU)$. Varying $\go$ had less effect (results not shown).

\section{Two finite entrance boundaries}
\label{sec:twoboundaries}
Many diffusions evolve on some interval $[a,b]$ with both $a$ and $b$ finite. Assuming $a$ to be an entrance boundary, the theory developed in this paper will apply only in the exceptional case that $\widetilde{\phi}(u)$ remains finite as $u \to b$. This will not typically hold if $b$ is an exit or an entrance boundary, for example, where we would expect $\ga(y) \to \pm\infty$ as $y \to b$; see, for example, the Wright-Fisher diffusion studied by \citet{sch:etal:2013}. While there is currently no solution available if we want to retain the Bessel process to simulate our candidate paths, one can make progress provided we revert to using Brownian motion. In this case we cannot make any assumptions on $[\ga^2 + \ga'](y)$ either as $y \to a$ or as $y \to b$; nonetheless, we can still use EA3 \citep[see \sref{sec:EA3} above and the Discussion in][]{bes:etal:2008}.

\begin{figure}[t]
\includegraphics[scale = 0.7]{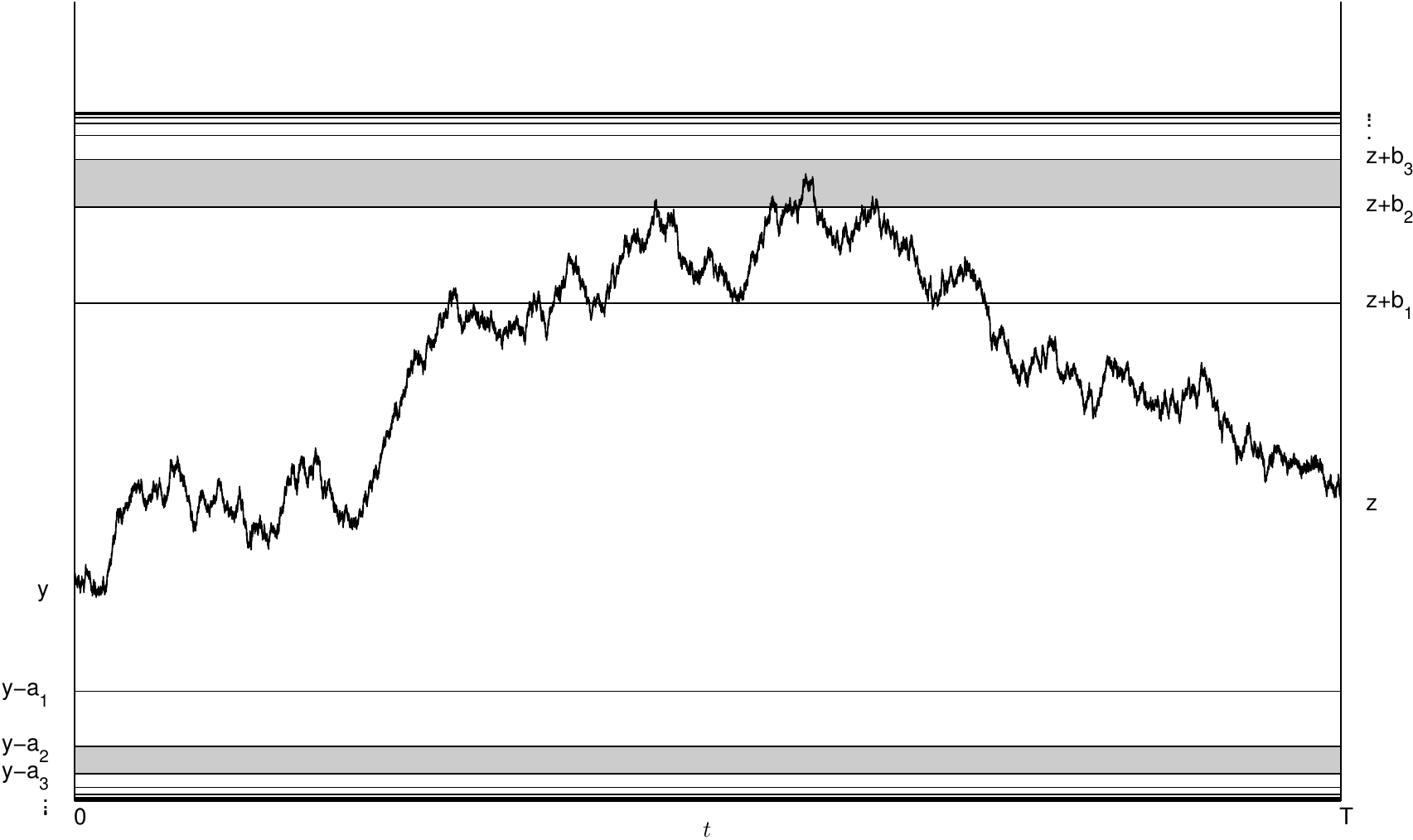}
\caption{\label{fig:layers}The layers of the state space for EA3. In this example, $Y_0 =y < z = Y_T$, and layers converge to each of the boundaries (both finite in this example). The layer of this sample path is $I = 3$, and $U_3$ has occurred.}
\end{figure}

The main idea is to partition $(a,b)$ into layers in such a way that the layers converge to each boundary (\fref{fig:layers}). As a consequence, simulated paths can approach either boundary arbitrarily closely, as required, yet we can still compute a bound $r(\gU)$ on $\phi$ whatever layer $I$ is actually simulated. Recall we must first simulate a layer $\gU = I$ of $Y$ (Step 2 of EA), and then simulate points of $Y$ given its layer (Step 4). A careful reading of the algorithm described by \citet{bes:etal:2008} shows that we need only modify their Step 4 to allow for unequal layers approaching the two boundaries, i.e.\ if EA3 is to be applied to two finite entrance boundaries, we can no longer choose to define the layers symmetrically, $a_i = b_i$, as is done in \citet{bes:etal:2008}. 
In the remainder of this section, I generalize their Step 4 to allow $a_i \neq b_i$.

Although it is not possible to simulate from $\bbW_y\mid Y_T, I$ directly, it \emph{is} possible to obtain points distributed according to this law by using a further rejection step, in which another, simpler process is used to sample candidate paths.  
Let $\bbW_{y \to z,T}$ denote the probability measure corresponding to a Brownian bridge from $y$ to $z$ in time $T$, and let $\bbW_{A; y_1\to y_2,T}$ denote the probability measure obtained after restricting the Brownian bridge to an event $A$. Define the events
\begin{align*}
\overline{M}_i &:= \left\{M_T \in [\zb + a_{i-1}, \zb + a_i)\right\},\\
\underline{M}_i &:= \left\{m_T \in [\yb - a_i, \yb - a_{i-1})\right\}.
\end{align*}
We sample candidate paths from the mixture
\begin{equation}
\label{eq:PDI}
\bbP_{D_I;y\to z,T} := \gl\bbW_{\overline{M}_I ; y\to z,T} + (1-\gl)\bbW_{\underline{M}_I ; y\to z,T}, \qquad \gl \in (0,1).
\end{equation}
The intuition behind this proposal measure is clear: it ensures that at least the first of the two sets constituting $U_I$ (or $L_I$) in equation \eref{eq:UI,LI} occurs. It is possible to simulate paths from \eref{eq:PDI} by first selecting from the mixture using a Bernoulli$(\gl)$ random variable, and then simulating a Brownian bridge conditioned on the chosen extremum falling within layer $I$. \citet{bes:etal:2006:B} provide details on how to simulate the maximum (or minimum) of a Brownian bridge and then simulate the rest of the bridge given this extremum. The distribution function of the extremum of a Brownian bridge is known in closed form, 
so it is a simple matter to condition the extremum to lie within a given interval. Rejection from $\bbP_{D_I;y\to z,T}$ to obtain paths distributed according to $\bbW_{D_I;y\to z,T}$ can be achieved by the following result. 
\begin{theorem}
\label{thm:EA3}
$\bbW_{D_I;y\to z,T}$ is absolutely continuous with respect to $\bbP_{D_I;y\to z,T}$, and if we choose
\begin{equation}
\label{eq:lambda}
\gl = \frac{\bbW_{y\to z,T}(\overline{M}_I)}{\bbW_{y\to z,T}(\overline{M}_I) + \bbW_{y\to z,T}(\underline{M}_I)},
\end{equation}
then the corresponding Radon-Nikod\'ym derivative is
\begin{equation}
\label{eq:thmEA3}
\frac{d\bbW_{D_I;y\to z,T}}{d\bbP_{D_I;y\to z,T}}(Y) = \frac{\bbW_{y\to z,T}(\overline{M}_I) + \bbW_{y\to z,T}(\underline{M}_I)}{\bbW_{y\to z,T}(\overline{D}_I)} \cdot \frac{\bbI(Y \in D_I)}{1 + \bbI(Y \in U_I \cap L_I)}.
\end{equation}
\end{theorem}
\begin{proof}
Use the unconditional Brownian bridge as a reference measure to find:
\begin{align*}
\ds\frac{d\bbW_{D_I;y\to z,T}}{d\bbP_{D_I;y\to z,T}}(Y) &= \frac{\ds\frac{d\bbW_{D_I;y\to z,T}}{d\bbW_{y\to z,T}}(Y)}{\gl\ds\frac{d\bbW_{\overline{M}_I;y\to z,T}}{d\bbW_{y\to z,T}}(Y) + (1-\gl)\frac{d\bbW_{\underline{M}_I;y\to z,T}}{d\bbW_{y\to z,T}} (Y)},\\
&= \frac{\bbI(Y \in D_I)\bbW_{y\to z,T}(D_I)^{-1}}{\gl\bbI(Y\in \overline{M}_I)\bbW_{y\to z,T}(\overline{M}_I)^{-1} + (1-\gl)\bbI(Y \in \underline{M}_I)\bbW_{y\to z,T}(\underline{M}_I)^{-1}},
\end{align*}
which simplifies to \eref{eq:thmEA3} when $\gl$ is given by \eref{eq:lambda} (also noting that
\[
\bbI(Y\in \overline{M}_I) + \bbI(Y\in \underline{M}_I) = 1 + \bbI(Y\in U_I\cap L_I)
\]
on the event $D_I$).
\end{proof}
Note that choosing $a_i = b_i$ recovers Theorem 4 of \citet{bes:etal:2008}, in which $\gl = 1/2$. Also notice that, because the distribution function of the extremum of a Brownian bridge is known, the general choice of $\gl$ in \eref{eq:lambda} can be computed exactly. It remains to simulate random indicators for the events $\{Y \in D_I\}$ (given $I$) and $\{Y \in U_I \cap L_I\}$ (given $I$ and $\{Y \in D_I\}$), which proceeds as in the symmetric case \citep{bes:etal:2008}. 

\section{Discussion}
\label{sec:discussion}
In this paper I have developed an efficient, exact algorithm for simulating from the law of a diffusion on $(a,\infty)$ with a finite entrance boundary at $a$. The algorithm is applicable when the boundary behaviour is matched by that of a Bessel process, 
which covers a number of interesting examples including conditioned diffusions [equation \eref{eq:conditioned}], the wide-sense Bessel process [equation \eref{eq:Besseldriftgen}], and a very general model of population growth [equation \eref{eq:growthgen}]. In an application to the latter, it was shown that using the Bessel process instead of Brownian motion to generate candidate paths gives a striking improvement in efficiency. For a diffusion with two entrance boundaries, I developed a tractable exact algorithm which uses Brownian motion as the candidate process.

There are a number of directions for further research. Perhaps the greatest restriction on the algorithm developed here is assumption (B*), which does not apply if for example the diffusion also has an upper entrance boundary. We should like to relax (B*) to the following:
\begin{itemize}
\item[(B**)] $\lim_{u \to 0+} \widetilde{\phi}(u) < \infty$.
\end{itemize}
This is the Bessel analogue of (**), and makes no restrictions on the drift away from $0$. 
For diffusions satisfying (**), \citet{bes:etal:2006:B} tackled the analogous problem by first simulating the maximum of a Brownian bridge path together with the time it is attained. This was possible because these distributions take on a simple form and are easy to simulate. There are grounds for optimism that we might take a similar approach using the Bessel process. Remarkably, the Bessel process is one of few well-characterized diffusions for which we also have some results on the distribution of the maximum of its bridge and the time it is attained \citep{pit:yor:1999, bor:sal:2002}. Here though, the relevant distributions are rather more complicated, expressible only in infinite series form. Exact simulation from these distributions will be the subject of a future paper.


A further extension of this work would be to handle other types of boundary behaviour. Much of the preceding argument, including \propref{prop:mak:gle:2010}, continues to hold for $0 < \gd < 2$, which could then be used in a rejection sampling algorithm for a target diffusion with an instantaneously reflecting boundary. However, great care must be taken in ensuring the assumptions of the algorithm are met. In particular we can no long write the Radon-Nikod\'ym derivative in the form \eref{eq:RDchain}; moreover, for $0 < \gd < 1$ the Bessel process is not even a semimartingale beyond $T_0$ \citep{rev:yor:1999, mij:uru:ip}.


Finally, the contributions in this paper illustrate an important concept: that it is possible to implement the exact algorithm using a \emph{non}-Brownian candidate process. This raises the interesting question: What candidate diffusions \emph{other than} Brownian motion and the Bessel process are available to use in the framework of the exact algorithm (and when would they be useful)?

\section*{Acknowledgements}
This work benefitted from many helpful discussions: with Steve Evans, Gareth Roberts, Joshua Schraiber, and Dario Span\`o.


\begin{thebibliography}{20}
\providecommand{\natexlab}[1]{#1}
\providecommand{\url}[1]{\texttt{#1}}
\expandafter\ifx\csname urlstyle\endcsname\relax
  \providecommand{\doi}[1]{doi: #1}\else
  \providecommand{\doi}{doi: \begingroup \urlstyle{rm}\Url}\fi

\bibitem[A\"it-Sahalia(2008)]{ait:2008}
A\"it-Sahalia, Y. (2008).
\newblock Closed-form likelihood expansions for multivariate diffusions.
\newblock \emph{Ann. Statist.}, {\bf 36}, 906--937.

\bibitem[Beskos and Roberts(2005)]{bes:rob:2005}
Beskos, A. and Roberts, G.~O. (2005).
\newblock Exact simulation of diffusions.
\newblock \emph{Ann. Appl. Probab.}, {\bf 15}, 2422--2444.

\bibitem[Beskos et~al.(2006{\natexlab{a}})Beskos, Papaspiliopoulos, and
  Roberts]{bes:etal:2006:B}
Beskos, A., Papaspiliopoulos, O., and Roberts, G.~O. (2006{\natexlab{a}}).
\newblock Retrospective exact simulation of diffusion sample paths with
  applications.
\newblock \emph{Bernoulli}, {\bf 12}, 1077--1098.

\bibitem[Beskos et~al.(2006{\natexlab{b}})Beskos, Papaspiliopoulos, Roberts,
  and Fearnhead]{bes:etal:2006:JRSSB}
Beskos, A., Papaspiliopoulos, O., Roberts, G.~O., and Fearnhead, P.
  (2006{\natexlab{b}}).
\newblock Exact and computationally efficient likelihood-based estimation for
  discretely observed diffusion processes.
\newblock \emph{J. R. Stat. Soc. Ser. B Stat. Methodol.}, {\bf 68}, 333--382.

\bibitem[Beskos et~al.(2008)Beskos, Papaspiliopoulos, and
  Roberts]{bes:etal:2008}
Beskos, A., Papaspiliopoulos, O., and Roberts, G.~O. (2008).
\newblock A factorisation of diffusion measure and finite sample path
  constructions.
\newblock \emph{Methodol. Comput. Appl. Probab.}, {\bf 10}, 85--104.

\bibitem[Borodin and Salminen(2002)]{bor:sal:2002}
Borodin, A.~N. and Salminen, P. (2002).
\newblock \emph{Handbook of {Brownian} motion: facts and formulae},
\newblock 2nd edition. Basel: Birkh\"auser Verlag.

\bibitem[Bullen(1998)]{bul:1998}
Bullen, P.~S. (1998).
\newblock \emph{A dictionary of inequalities}, volume~97 of \emph{Pitman
  monographs and surveys in pure and applied mathematics}.
\newblock Harlow: Longman.

\bibitem[Fearnhead et~al.(2008)Fearnhead, Papaspiliopoulos, and
  Roberts]{fea:etal:2008}
Fearnhead, P., Papaspiliopoulos, O., and Roberts, G.~O. (2008).
\newblock Particle filters for partially observed diffusions.
\newblock \emph{J. R. Stat. Soc. Ser. B Stat. Methodol.}, {\bf 40}, 755--777.

\bibitem[Fearnhead et~al.(2010)Fearnhead, Papaspiliopoulos, Roberts, and
  Stuart]{fea:etal:2010}
Fearnhead, P., Papaspiliopoulos, O., Roberts, G.~O., and Stuart, A. (2010).
\newblock Random-weight particle filtering of continuous time processes.
\newblock \emph{J. R. Stat. Soc. Ser. B Stat. Methodol.}, {\bf 72}, 497--512.

\bibitem[Karlin and Taylor(1981)]{kar:tay:1981}
Karlin, S. and Taylor, H.~M. (1981).
\newblock \emph{A second course in stochastic processes}.
\newblock New York: Academic Press.

\bibitem[Kloeden and Platen(1999)]{klo:pla:1999}
Kloeden, P.~E. and Platen, E. (1999).
\newblock \emph{Numerical solution of stochastic differential equations},
\newblock 3rd printing. Berlin: Springer-Verlag.

\bibitem[Makarov and Glew(2010)]{mak:gle:2010}
Makarov, R.~N. and Glew, D. (2010).
\newblock Exact simulation of {Bessel} diffusions.
\newblock \emph{Monte Carlo Methods Appl.}, {\bf 16}, 286--306.

\bibitem[Mijatovi\'c and Urusov(2013)]{mij:uru:ip}
Mijatovi\'c, A. and Urusov, M. (2013).
\newblock On the loss of the semimartingale property at the hitting time of a
  level.
\newblock \emph{J. Theoret. Probab.}
\newblock In press. {\it arXiv:1304.1377}.

\bibitem[Pitman and Yor(1981)]{pit:yor:1981}
Pitman, J. and Yor, M. (1981).
\newblock Bessel processes and infinitely divisible laws.
\newblock In Williams, D., editor, \emph{Stochastic Integrals}, volume 851 of
  \emph{Lecture Notes in Mathematics}, pages 285--370. Berlin: Springer.

\bibitem[Pitman and Yor(1999)]{pit:yor:1999}
Pitman, J. and Yor, M. (1999).
\newblock The law of the maximum of a {Bessel} bridge.
\newblock \emph{Electron. J. Probab.}, {\bf 4}, No.~15, 1--35.

\bibitem[Revuz and Yor(1999)]{rev:yor:1999}
Revuz, D. and Yor, M. (1999).
\newblock \emph{Continuous martingales and {Brownian} motion},
\newblock 3rd edition. Berlin: Springer-Verlag.

\bibitem[Rogers and Pitman(1981)]{rog:pit:1981}
Rogers, L. C.~G. and Pitman, J.~W. (1981).
\newblock Markov functions.
\newblock \emph{Ann. Probab.}, {\bf 9}, 573--582.

\bibitem[Salminen(1984)]{sal:1984}
Salminen, P. (1984).
\newblock One-dimensional diffusions and their exit spaces.
\newblock \emph{Math. Scand.}, {\bf 54}, 209--220.

\bibitem[Schraiber et~al.(2013)Schraiber, Griffiths, and Evans]{sch:etal:2013}
Schraiber, J., Griffiths, R.~C., and Evans, S.~N. (2013).
\newblock Analysis and rejection sampling of {Wright-Fisher} diffusion bridges.
\newblock \emph{Theoret. Population Biology}, {\bf 89}, 64--74.

\bibitem[Watanabe(1975)]{wat:1975:ZWVG}
Watanabe, S. (1975).
\newblock On time inversion of one-dimensional diffusion processes.
\newblock \emph{Z. Wahrscheinlichkeitstheorie und Verw. Gebiete}, {\bf 31},
  115--124.

\end{thebibliography}
\end{document}